\documentclass[letterpaper,journal]{IEEEtran}

\usepackage{amsmath,amsfonts,amssymb}
\usepackage{algorithmic}
\usepackage{algorithm}
\usepackage{graphicx}
\usepackage{booktabs}
\usepackage{cite}
\usepackage{orcidlink}
\usepackage{balance}
\hypersetup{hidelinks}

\newtheorem{theorem}{Theorem}

\begin{document}

\title{Distributed Power Control for Equal-Priority NGSO Mega-Constellation Coexistence}

\author{Hao~Yuan\,\orcidlink{0000-0002-5670-6171},~Shuang~Zheng\,\orcidlink{0000-0002-3090-4196},~Heying~Cao\,\orcidlink{0009-0006-6008-6143},~and~Xing~Zhang\,\orcidlink{0000-0003-4345-6166},~\IEEEmembership{Senior~Member,~IEEE}%
\thanks{Hao Yuan, Heying Cao, and Xing Zhang are with the School of Information and Communication Engineering, Beijing University of Posts and Telecommunications, Beijing 100876, China (e-mail: yuan\_hao@bupt.edu.cn; hszhang@bupt.edu.cn).}%
\thanks{Shuang Zheng is with the School of Electronic and Information Engineering, Tiangong University, Tianjin 300387, China (e-mail: zshuang@tiangong.edu.cn).}%
}

\markboth{IEEE Communications Letters}%
{H.~Yuan \MakeLowercase{\textit{et al.}}: Distributed Power Control for Equal-Priority NGSO Mega-Constellation Coexistence}

\maketitle


\begin{abstract}
Two non-geostationary orbit (NGSO) mega-constellations over shared spectrum leads to inter-constellation interference. Existing mitigation mainly relies on frequency/time partitioning, spatial isolation and satellite selection. However, these approaches may sacrifice resource reuse or be constrained by spatial geometry. To enable interference mitigation without centralized coordination, this letter investigates distributed power control for independent, equal-priority NGSO operators. We formulate the interaction as a non-cooperative game and propose a distributed water-filling best response using local serving-channel estimates and aggregate interference measurements. Pure-strategy Nash equilibrium existence and a weighted spectral-norm sufficient condition for uniqueness and convergence are established. Simulation results demonstrate utility gains of 3.91\% over uncoordinated transmission with 20 active beams.
\end{abstract}


\section{Introduction}
\IEEEPARstart{T}{he} rapid deployment of non-geostationary orbit (NGSO) mega-constellations, such as SpaceX's Starlink and Amazon's Kuiper, has intensified co-frequency interference among independent satellite operators. The ITU Radio Regulations impose equivalent power flux-density limits~\cite{ITU2024}, whereas FCC rules specify NGSO fixed-satellite-service sharing procedures~\cite{FCC2024}. Recent work surveys inter-constellation interference and avoidance~\cite{CoFreq2023} and synthesizes FCC/ITU regulations with coexistence mechanisms spanning frequency, time, power, and space~\cite{Roberts2026}. For direct coexistence, radio-resource management for multiple large constellations was studied in~\cite{Re2021}, while probabilistic look-aside using ephemeris information with limited real-time information exchange was developed in~\cite{Lin2026}. These approaches differ from an operator-level continuous power game with equilibrium guarantees.

Game theory provides a natural framework for strategic interactions in satellite networks~\cite{Jiang2024}. Related studies address satellite--terrestrial coexistence~\cite{Chen2021}, GEO--LEO interference mitigation~\cite{Jalali2023,Gu2022}, potential-game resource allocation within one satellite network~\cite{Zhang2021}, distributed channel allocation for a GEO--LEO hybrid constellation~\cite{Li2024Game}, and joint beam scheduling and power control in a mega hybrid constellation~\cite{Li2025IoT}. In contrast, we model two same-tier, independent NGSO operators as strategic players. Each selects a continuous per-beam power vector under its own budget and utility, without an incumbent--secondary hierarchy or a centralized joint controller. The distinction is therefore the operator-level, equal-priority NGSO--NGSO game, rather than game-theoretic satellite power control per se.

This gap motivates a distributed framework that connects the coexistence model, equilibrium theory, and measurable implementation requirements. Our contributions are:
\begin{itemize}
\setlength{\itemsep}{0pt}
\setlength{\parskip}{0pt}
\setlength{\parsep}{0pt}
\item \emph{Scenario and modeling:} Cross-constellation power control between two independent, equal-priority NGSO operators sharing a downlink band is formulated as a non-cooperative game with per-beam power vectors and operator-local constraints.
\item \emph{Theory:} A potential-like characterization with an explicit additive weak-coupling error is established. Pure-strategy NE existence is proved, together with a weighted spectral-norm sufficient condition for uniqueness and convergence under cumulative cross-constellation coupling.
\item \emph{Algorithm and evidence:} A distributed water-filling best response requiring only local serving-channel estimates and aggregate interference measurements is derived. Numerical results evaluate convergence, scalability, geometric sensitivity, and the efficiency--lower-tail trade-off against uncoordinated, heuristic, max-min, and multi-start centralized benchmarks.
\end{itemize}

\section{System Model}

Consider two NGSO operators $\mathcal{A}$ and $\mathcal{B}$ serving a shared Ku-band area (Fig.~\ref{fig:system_scenario}). Operator $\mathcal{A}$ uses altitude $a_{\mathcal{A}} = 550$~km and $\mathcal{B}$ uses $a_{\mathcal{B}} = 630$~km. On one 250~MHz time--frequency resource-block snapshot, operator $i$ activates $K_i$ spot beams and schedules one terminal per beam. All active beams of both operators reuse the block. To isolate inter-operator coexistence, intra-operator inter-beam interference is assumed ideally suppressed by precoding/coordination; hence all $K_{\bar i}$ beams of the other operator $\bar i$ contribute to~\eqref{eq:sinr}.

\begin{figure}[!t]
\centering
\includegraphics[width=0.86\columnwidth]{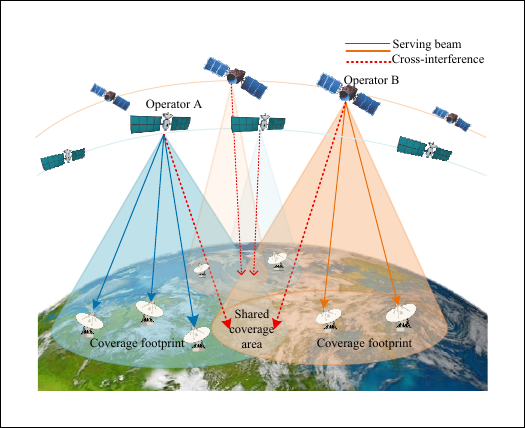}
\caption{Two independent NGSO operators reuse a downlink resource block. Blue/orange solid links denote service, and red dashed links denote cross-operator interference; orbital tracks and beam counts are schematic.}
\label{fig:system_scenario}
\end{figure}

\subsection{Signal Model}

For operator $i$'s beam $k$, the received SINR is given by~\eqref{eq:sinr},
\begin{equation}\label{eq:sinr}
\gamma_k^{i} = \frac{p_k^{i} g_k^{i}}{\sigma^2 + I_k^{\bar i}}, \quad I_k^{\bar i} = \sum_{j=1}^{K_{\bar i}} p_j^{\bar i} h_{kj}^{i},
\end{equation}
where $p_k^{i} \geq 0$ is the transmit power; $g_k^{i}$ is the serving channel gain including free-space path loss and antenna gains; $h_{kj}^{i}$ is the cross-interference gain from beam $j$ of operator $\bar i$ to user $k$ of operator $i$; $\sigma^2 = k_{\mathrm{B}} T B$ is the thermal noise power with noise temperature $T = 300$~K and bandwidth $B = 250$~MHz; and $I_k^{\bar i}$ is the aggregate cross-interference.

\subsection{Interference Channel Model}

The serving gain $g_k^{i}$ follows the free-space path loss model with satellite and terminal antenna gains. For a 0.6-m parabolic user terminal, the receive gain toward an interfering satellite at off-axis angle $\phi$ (degrees) follows the ITU-R S.1428-1 reference pattern~\cite{ITU1428}. With $q=D/\lambda$ and $20\le q\le25$, it is
\begin{equation}\label{eq:sidelobe}
G_r(\phi) = \begin{cases}
G_{\max}-2.5\!\times\!10^{-3}(q\phi)^2,&0\le\phi<\phi_m,\\
G_1,&\phi_m\le\phi<95/q,\\
29-25\log_{10}\phi,&95/q\le\phi\le33.1^\circ,\\
-9,&33.1^\circ<\phi\le80^\circ,\\
-5,&80^\circ<\phi\le180^\circ,
\end{cases}
\end{equation}
where $G_{\max}=20\log_{10}q+7.7$, $G_1=29-25\log_{10}(95/q)$, and $\phi_m=(20/q)\sqrt{G_{\max}-G_1}$. At the interfering satellite, $\psi_{kj}^{i}$ is the angle between beam $j$'s boresight toward its own scheduled user and the direction toward victim user $k$. With no measured pattern available, the LEO branch of ITU-R S.1528-1~\cite{ITU1528} gives
\begin{equation}\label{eq:sat_pattern}
G_s(\psi)=\begin{cases}
G_{s,\max}-3(\psi/\psi_b)^2,&0\le\psi\le Y,\\
G_{s,\max}+L_s-25\log_{10}\!(\psi/Y),&Y<\psi\le Z,\\
L_F,&Z<\psi\le180^\circ,
\end{cases}
\end{equation}
where $L_s=-6.75$~dB, $Y=1.5\psi_b$, and $Z=Y10^{0.04(G_{s,\max}+L_s-L_F)}$. We use $G_{s,\max}=30$~dBi, infer $D_s/\lambda=12.49$ with 65\% aperture efficiency, set $\psi_b=\sqrt{1200}/(D_s/\lambda)=2.77^\circ$, and use the conservative Annex-1 value $L_F=5$~dBi. The dual-ended cross gain is
\begin{equation}\label{eq:cross_gain}
h_{kj}^{i}=10^{[G_s(\psi_{kj}^{i})+G_r(\phi_{kj}^{i})-L_{\rm FS}(d_{kj}^{i})]/10}.
\end{equation}

\subsection{Utility and Strategy}

Each operator independently maximizes the weighted sum spectral-efficiency (weighted sum-rate) utility
\begin{equation}\label{eq:utility}
u_i(\mathbf{p}_i, \mathbf{p}_{\bar i}) = \sum_{k=1}^{K_i} w_k^{i} \log_2\!\bigl(1 + \gamma_k^{i}\bigr),
\end{equation}
over the strategy set $\mathcal{S}_i = \{\mathbf{p}_i : \sum_k p_k^{i} \leq P_{\mathrm{tot}}^{i},\ p_{\min} \leq p_k^{i} \leq p_{\max},\ \forall k\}$, where $P_{\mathrm{tot}}^{i}$ is the total power budget. We assume $K_i p_{\min}\le P_{\mathrm{tot}}^{i}$; if $P_{\mathrm{tot}}^{i}>K_i p_{\max}$, the budget is inactive and the box upper bounds limit the usable power. The objective maximizes aggregate spectral efficiency; it does not impose proportional fairness, minimum-rate, or outage guarantees.

\section{Nash Equilibrium Analysis and Distributed Solution}\label{sec:analysis}

Define the game $\mathcal{G} = (\mathcal{N}, \{\mathcal{S}_i\}, \{u_i\})$ with $\mathcal{N} = \{\mathcal{A}, \mathcal{B}\}$.

\subsection{Potential Structure}

We first characterize how cross-interference perturbs the separable structure each operator would enjoy in isolation.

\begin{theorem}\label{thm:potential}
Define the potential-like function, normalized by an arbitrary reference power $P_0$,
\begin{equation}\label{eq:potential}
\Phi(\mathbf{p}) = \sum_{i \in \mathcal{N}} \sum_{k=1}^{K_i} w_k^{i} \log_2\!\frac{p_k^{i} g_k^{i} + \sigma^2 + I_k^{\bar i}}{P_0},
\end{equation}
and the weak-interference ratio
\begin{equation}\label{eq:weak}
\eta \triangleq \frac{K \cdot \max_{k,j,i} h_{kj}^{i}}{\min_{k,i} g_k^{i}}, \quad K = \max(K_{\mathcal{A}}, K_{\mathcal{B}}).
\end{equation}
Then (i) for any fixed unilateral deviation, $\Delta\Phi-\Delta u_i\!\to\!0$ as $h_{kj}^{i}\!\to\!0$ (with equality at zero cross-interference); and (ii) in general,
\begin{equation}\label{eq:approx}
|\Delta\Phi-\Delta u_i|\le
\frac{K_iK_{\bar i}\bar w\,\Delta p_{\max}\bar h}{(\ln2)p_{\min}g_{\min}}
\triangleq\epsilon_{\Phi},
\end{equation}
where $\bar h=\max h_{kj}^{i}$. For fixed $K_{\mathcal A}$ and $K_{\mathcal B}$, $\epsilon_{\Phi}=O(\eta)$ as $\bar h/g_{\min}\to0$. This additive bound does not by itself imply ordinal sign preservation when the utility change is of the same order as the error.
\end{theorem}

\begin{IEEEproof}
Write $N_k^{i}=p_k^{i}g_k^{i}+\sigma^{2}+I_k^{\bar i}$ and $D_k^{i}=\sigma^{2}+I_k^{\bar i}$. The normalization by $P_0$ cancels in every difference, so $u_i=\sum_k w_k^{i}(\log_2 N_k^{i}-\log_2 D_k^{i})$ and
\begin{equation}\label{eq:phi_expand}
\Phi=\sum_{i\in\mathcal N}u_i+\underbrace{\sum_{i,k}w_k^{i}\log_2(D_k^{i}/P_0)}_{\triangleq\, C(\mathbf p)}.
\end{equation}
Consider operator $\mathcal{A}$ changing $\mathbf{p}_{\mathcal{A}}\to\mathbf{p}_{\mathcal{A}}'$ with $\mathbf{p}_{\mathcal{B}}$ fixed. Since $D_k^{\mathcal{A}}=\sigma^{2}+I_k^{\mathcal{B}}$ depends on $\mathbf{p}_{\mathcal{B}}$ only, the $-\log_2 D_k^{\mathcal{A}}$ terms in $u_{\mathcal{A}}$ are invariant. For operator $\mathcal{B}$, both $N_k^{\mathcal{B}}$ and $D_k^{\mathcal{B}}=\sigma^{2}+I_k^{\mathcal{A}}$ vary, but the $-\log_2 D_k^{\mathcal{B}}$ part of $\Delta u_{\mathcal{B}}$ cancels exactly with the matching term in $\Delta C$. Hence
\begin{equation}\label{eq:dphi}
\Delta\Phi = \Delta u_{\mathcal{A}} + \Delta R,
\end{equation}
with the single cross-term
\begin{equation}\label{eq:cross}
\begin{split}
\Delta R &= \sum_{k=1}^{K_{\mathcal{B}}} w_k^{\mathcal{B}}\log_2\!\frac{p_k^{\mathcal{B}}g_k^{\mathcal{B}}+D_k^{\mathcal{B}\prime}}{p_k^{\mathcal{B}}g_k^{\mathcal{B}}+D_k^{\mathcal{B}}} \\
&= \sum_{k} w_k^{\mathcal{B}}\log_2\!\Bigl(1+\tfrac{\Delta I_k^{\mathcal{A}}}{p_k^{\mathcal{B}}g_k^{\mathcal{B}}+D_k^{\mathcal{B}}}\Bigr),
\end{split}
\end{equation}
where $D_k^{\mathcal{B}\prime}=\sigma^{2}+I_k^{\mathcal{A}}(\mathbf{p}_{\mathcal{A}}')$ and $\Delta I_k^{\mathcal{A}}=I_k^{\mathcal{A}}(\mathbf{p}_{\mathcal{A}}')-I_k^{\mathcal{A}}(\mathbf{p}_{\mathcal{A}})$. (i) As $h_{kj}^{i}\!\to\!0$, $\Delta I_k^{\mathcal{A}}\!\to\!0$; thus $\Delta\Phi-\Delta u_{\mathcal{A}}=\Delta R\!\to\!0$. At zero cross-coupling, equality holds and $\Phi$ is an exact (thus ordinal) potential~\cite{Monderer1996}. (ii) For positive or negative $\Delta I_k^{\mathcal{A}}$, the mean-value theorem gives $|\log_2 x-\log_2 y|\le |x-y|/(\ln2\,\min\{x,y\})$. Here both arguments are at least $p_{\min}g_{\min}$. Also, $|\Delta I_k^{\mathcal{A}}|\le K_{\mathcal{A}}\Delta p_{\max}\bar h$ with $\Delta p_{\max}=p_{\max}-p_{\min}$ and $\bar h=\max_{k,j}h_{kj}$, hence
\begin{equation}\label{eq:dbound}
|\Delta R|\le \frac{K_{\mathcal{A}}K_{\mathcal{B}}\,\bar w\,\Delta p_{\max}\,\bar h}{(\ln 2)\,p_{\min}g_{\min}},
\end{equation}
with $\bar w=\max_{k,i}w_k^{i}$. For fixed beam counts this bound is $O(\eta)$ because $\eta=K\bar h/g_{\min}$. Therefore $|\Delta\Phi-\Delta u_{\mathcal{A}}|=|\Delta R|\le\epsilon_{\Phi}$.
\end{IEEEproof}

\textit{Remark~1 (scope):} The weak-coupling result in Theorem~\ref{thm:potential} is asymptotic and does not characterize the experimental regime; Section~\ref{sec:results} instead assesses convergence using the sufficient contraction diagnostic, direct iterations, and best-response residuals.

\subsection{Existence}

\begin{theorem}\label{thm:existence}
$\mathcal{G}$ admits at least one pure-strategy NE.
\end{theorem}

\begin{IEEEproof}
Each $\mathcal{S}_i$ is compact, convex, and non-empty, and $u_i$ is continuous and concave (hence quasi-concave) in $\mathbf{p}_i$ for fixed $\mathbf{p}_{\bar i}$, since every summand $\log_2(1+p_k^{i}g_k^{i}/D_k^{i})$ is concave and increasing in $p_k^{i}$. Existence of a pure-strategy NE follows from Rosen's concave-game existence result~\cite{Rosen1965}.
\end{IEEEproof}

\subsection{Distributed Water-Filling Best Response}

For fixed $\mathbf p_{\bar i}$, operator $i$ solves a convex best-response problem. Since~\eqref{eq:utility} is increasing in every $p_k^i$, the active total power is $P_{\mathrm{use}}^i=\min(P_{\mathrm{tot}}^i,K_i p_{\max})$; for $K_i p_{\max}<P_{\mathrm{tot}}^i$, all beams saturate at $p_{\max}$. Otherwise, the KKT condition
\[
\frac{w_k^i g_k^i}{\ln 2\,(\sigma^2+I_k^{\bar i}+p_k^i g_k^i)}=\nu_i
\]
gives, with $\lambda_i=\nu_i\ln2$,
\begin{equation}\label{eq:wf}
p_k^{i*} = \left[\frac{w_k^{i}}{\lambda_i} - \frac{\sigma^2 + I_k^{\bar i}}{g_k^{i}}\right]_{p_{\min}}^{p_{\max}},
\end{equation}
where $\lambda_i$ is found by bisection with complexity $\mathcal{O}(K\log(1/\delta))$ per operator per iteration.

The upper bisection bound in Algorithm~\ref{alg:br} makes the clipped solution equal to $p_{\min}$ for all beams. The algorithm requires only local serving-channel estimates and aggregate interference measurements; the full cross-channel matrix is used only in the offline simulation. The following theorem gives a sufficient convergence condition.

\begin{theorem}\label{thm:convergence}
Let $F_{i\bar i}=\mathrm{diag}(\mathbf g_i)^{-1}H_{i\bar i}$, where $H_{i\bar i}[k,j]=h_{kj}^{i}$, and let $W_i=\mathrm{diag}(\mathbf w_i)\succ0$. Define
\begin{equation}\label{eq:jsum}
L_{i\bar i}=\|W_i^{-1/2}F_{i\bar i}W_{\bar i}^{1/2}\|_2,\quad
J_2=\begin{bmatrix}0&L_{\mathcal A\mathcal B}\\L_{\mathcal B\mathcal A}&0\end{bmatrix}.
\end{equation}
If $\rho(J_2)=\sqrt{L_{\mathcal A\mathcal B}L_{\mathcal B\mathcal A}}<1$, the game has a unique NE, and both simultaneous and alternating water-filling best-response iterations converge to it geometrically.
\end{theorem}

\begin{IEEEproof}
Let $\mathcal X_i=\{\mathbf p_i:p_{\min}\le p_k^i\le p_{\max},\ \mathbf1^T\mathbf p_i=P_{\mathrm{use}}^i\}$ and $\mathbf n_i=(\sigma^2\mathbf1+H_{i\bar i}\mathbf p_{\bar i})\oslash\mathbf g_i$. The KKT solution in~\eqref{eq:wf} is equivalently
\[
\begin{aligned}
B_i(\mathbf n_i)&=\arg\min_{\mathbf p_i\in\mathcal X_i}
\tfrac12\|\mathbf p_i+\mathbf n_i\|_{W_i^{-1}}^2,\\
\|\mathbf x\|_{W_i^{-1}}&\triangleq\|W_i^{-1/2}\mathbf x\|_2.
\end{aligned}
\]
Hence $B_i$ is a metric projection and is non-expansive. Therefore $\|\Delta B_i\|_{W_i^{-1}}\le\|\Delta\mathbf n_i\|_{W_i^{-1}}\le L_{i\bar i}\|\Delta\mathbf p_{\bar i}\|_{W_{\bar i}^{-1}}$. The two block errors are dominated by $J_2$. If $\rho(J_2)<1$, a weighted block-maximum norm makes the simultaneous map contractive, so Banach's theorem gives uniqueness and geometric convergence. For alternating updates, each full sweep contracts by at most $L_{\mathcal A\mathcal B}L_{\mathcal B\mathcal A}=\rho(J_2)^2<1$.
\end{IEEEproof}

\section{Numerical Results}\label{sec:results}

\begin{figure}[!t]
\centering
\includegraphics[width=\columnwidth]{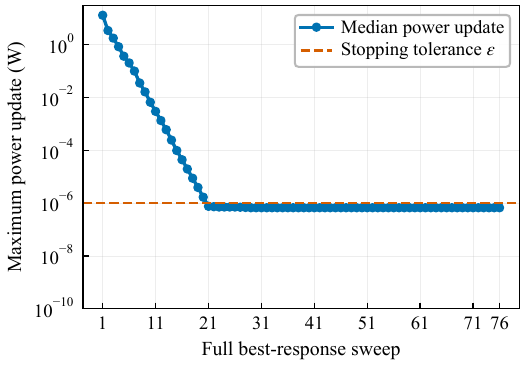}
\caption{Near-inline convergence of Algorithm~\ref{alg:br} for $K=20$ over 50 realizations. The solid curve is the median of $\max_i\|\mathbf p_i^{(t)}-\mathbf p_i^{(t-1)}\|_2$, and the dashed line marks $\epsilon=10^{-6}$~W.}
\label{fig:convergence}
\end{figure}

\begin{figure}[!t]
\centering
\includegraphics[width=\columnwidth]{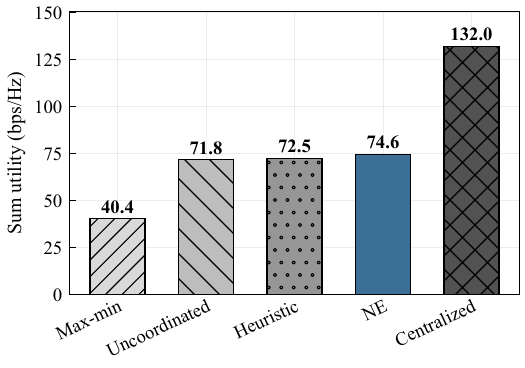}
\caption{Near-inline aggregate utility for $K=20$ over 50 realizations. The NE improves on uncoordinated transmission by 3.91\%; the heuristic is a simple-rule reference, and the centralized result is a multi-start benchmark.}
\label{fig:utility}
\end{figure}

\begin{figure}[!t]
\centering
\includegraphics[width=\columnwidth]{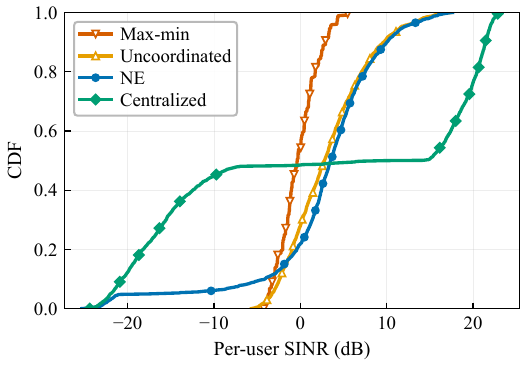}
\caption{Near-inline empirical CDF of per-user SINR for $K=20$ over 50 realizations.}
\label{fig:sinr_cdf}
\end{figure}

We evaluate the framework with Table~\ref{tab:params}. Free-space path loss at 11.7~GHz over 550~km is 168.6~dB and noise power is $-119.8$~dBW, giving 29.3~dB interference-free SINR for 43~dBW EIRP and 35.1~dBi terminal gain.

\begin{table}[!t]
\caption{Simulation Parameters}\label{tab:params}
\centering
\begin{tabular}{lc}
\toprule
\textbf{Parameter} & \textbf{Value} \\
\midrule
Constellation $\mathcal{A}$ altitude & 550~km \\
Constellation $\mathcal{B}$ altitude & 630~km \\
Carrier frequency & 11.7~GHz \\
Bandwidth per beam & 250~MHz \\
Total power budget & 150~W per operator \\
Max/min power per beam & 20~W / 0.1~W \\
Satellite peak gain & 30~dBi \\
Satellite pattern / half 3-dB beamwidth & ITU-R S.1528-1 / $2.77^{\circ}$ \\
Terminal diameter / peak gain & 0.6~m / 35.1~dBi \\
Noise temperature & 300~K \\
Users per operator & 20 \\
Service region / co-cell jitter radius & 100~km / 10~km \\
Nominal elevation / azimuth offset & $45$--$80^{\circ}$ / $\pm5^{\circ}$ \\
Nominal receiver-angle floor & $3^{\circ}$ \\
Near-inline elevation / separation & $55$--$65^{\circ}$ / $2$--$3^{\circ}$ \\
\bottomrule
\end{tabular}
\end{table}

\begin{algorithm}[!b]
\caption{Distributed Water-Filling Best Response}\label{alg:br}
\begin{algorithmic}[1]
\STATE \textbf{Input:} $\{g_k^{i}\}$, $\{w_k^{i}\}$, $\sigma^2$, $p_{\min}$, $p_{\max}$, $P_{\mathrm{tot}}^{i}$, $\epsilon$, $\delta$
\STATE $P_{\mathrm{use}}^{i}\leftarrow \min(P_{\mathrm{tot}}^{i},K_i p_{\max})$
\STATE $p_k^{i} \leftarrow P_{\mathrm{use}}^{i}/K_i$ for all $i,k$ \COMMENT{Feasible equal-power init.}
\STATE $t \leftarrow 0$
\REPEAT
\FOR{each operator $i \in \{\mathcal{A}, \mathcal{B}\}$ (alternating)}
\STATE Measure $I_k^{\bar i}$ for $k=1,\ldots,K_i$ \COMMENT{Aggregate cross-interf.}
\IF{$P_{\mathrm{use}}^{i}\ge K_i p_{\max}-\delta$}
\STATE $p_k^{i}\leftarrow p_{\max}$ for all $k$; \textbf{continue}
\ENDIF
\STATE $\lambda_{\mathrm{lo}} \leftarrow 0$;\; $\lambda_{\mathrm{hi}} \leftarrow \max_k \frac{w_k^{i} g_k^{i}}{\sigma^2 + I_k^{\bar i}+p_{\min}g_k^{i}}$
\REPEAT
\STATE $\lambda \leftarrow (\lambda_{\mathrm{lo}}+\lambda_{\mathrm{hi}})/2$
\FOR{$k = 1,\ldots,K_i$}
\STATE $\tilde{p}_k \leftarrow w_k^{i}/\lambda - (\sigma^2 + I_k^{\bar i})/g_k^{i}$
\STATE $p_k^{i} \leftarrow [\tilde{p}_k]_{p_{\min}}^{p_{\max}}$ \COMMENT{Project to box}
\ENDFOR
\IF{$\sum_k p_k^{i} > P_{\mathrm{use}}^{i}$}
\STATE $\lambda_{\mathrm{lo}} \leftarrow \lambda$
\ELSE
\STATE $\lambda_{\mathrm{hi}} \leftarrow \lambda$
\ENDIF
\UNTIL{$|\sum_k p_k^{i} - P_{\mathrm{use}}^{i}| < \delta$}
\ENDFOR
\STATE $t \leftarrow t + 1$
\UNTIL{$\max_i \|\mathbf{p}^{i,(t)} - \mathbf{p}^{i,(t-1)}\|_2 < \epsilon$}
\STATE \textbf{Output:} $\{p_k^{i,*}\}$
\end{algorithmic}
\end{algorithm}

\textbf{Baseline Definitions.} Uncoordinated transmission water-fills against noise alone, ignoring cross-operator interference. The interference-aware heuristic sets $p_k^i\propto[1+(\mathbf H_{i\bar i}\mathbf p_{\bar i})_k/(g_k^i p_{\mathrm{nom}}^i+\sigma^2)]^{-1}$, where $p_{\mathrm{nom}}^i=\min(p_{\max},P_{\mathrm{tot}}^i/K_i)$, then clips and renormalizes over two alternating rounds. Max-min alternates best responses that maximize the minimum beam SINR. NE is the converged Algorithm~\ref{alg:br} profile; the multi-start centralized benchmark jointly maximizes $u_{\mathcal A}+u_{\mathcal B}$ without claiming global optimality. All schemes share channel realizations and power constraints and are scored in Fig.~\ref{fig:utility} by the weighted sum spectral efficiency $u_{\mathcal A}+u_{\mathcal B}$ (unit weights).

Fifty realizations (seed 42) use spherical-Earth geometry. Each operator's multibeam satellite serves co-located cells in a 100-km-radius region at $40^\circ$N; the two scheduled users per cell are independently jittered within 10~km. Beam $j$ points to user $j$, so the three-dimensional geometry determines both $\phi_{kj}^{i}$ and $\psi_{kj}^{i}$. The \emph{nominal} scenario samples Table~\ref{tab:params}'s elevation and azimuth ranges; its $3^\circ$ receive-angle floor represents alternative-satellite selection in near-inline geometry~\cite{ITU1431}. The \emph{near-inline} scenario instead conditions center-look separation uniformly on $[2^\circ,3^\circ]$ and uses the raw receive angle, representing the window before avoidance. Centralized SLSQP uses seven starts, maxiter 2000, and ftol $10^{-12}$.

\textbf{Nominal Reference.} The nominal $h_{kj}^i/g_k^i$ median, 95th percentile, and maximum are $10^{-4}$, 0.0179, and 0.0774. All 50 runs satisfy Theorem~\ref{thm:convergence} and converge; mean (maximum) $\rho(J_2)$ is 0.076 (0.337). The NE gains 0.06\% over uncoordinated transmission and reaches 95.7\% of the multi-start benchmark. Satellite selection thus leaves the nominal case weakly coupled with little adaptation gain.

\textbf{Near-Inline Convergence.} The near-inline $h/g$ median, 95th percentile, and maximum are 0.0094, 0.1850, and 0.9920. Fig.~\ref{fig:convergence} reports the stopping metric. Here $r_{\mathrm{BR},\infty}\triangleq\max_i\|\mathbf p_i-B_i(\mathbf p_{\bar i})\|_\infty$ is the terminal best-response residual, with $B_i(\mathbf p_{\bar i})$ denoting the best response to the other profile. All 50 runs converge in 9--76 sweeps (mean 24.48), with $r_{\mathrm{BR},\infty}<4.5\!\times\!10^{-7}$~W. The sufficient condition holds for 19/50 runs, with mean $\rho(J_2)=1.201$; convergence outside it is empirical. Since mean $\eta=13.195$, the additive potential approximation is not asserted tight here.

\textbf{Near-Inline Performance.} In Fig.~\ref{fig:utility}, the NE gains 3.91\% over uncoordinated transmission and reaches 56.5\% of the centralized benchmark. Fig.~\ref{fig:sinr_cdf}'s 5th-percentile SINRs are $-3.29$~dB (uncoordinated), $-17.08$~dB (NE), $-3.60$~dB (max-min), and $-21.99$~dB (centralized). Sum-utility optimization sacrifices the pooled lower tail; max-min preserves a tail similar to uncoordinated but forfeits 45.9\% of the NE utility.

\begin{figure}[!t]
\centering
\includegraphics[width=\columnwidth]{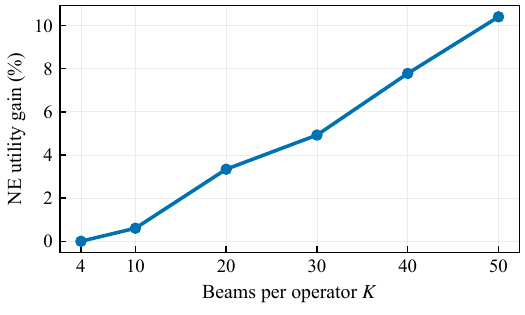}
\caption{Near-inline NE utility gain over uncoordinated transmission versus beams per operator; each point averages 20 realizations.}
\label{fig:scalability}
\end{figure}

\begin{figure}[!t]
\centering
\includegraphics[width=\columnwidth]{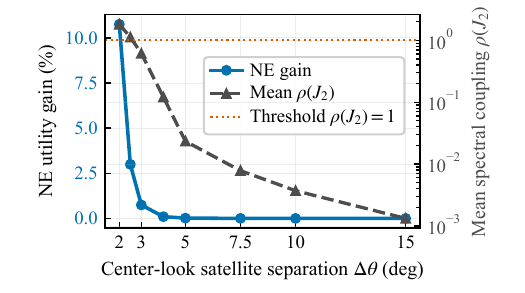}
\caption{Physical separation sweep for $K=20$; NE utility gain is relative to uncoordinated transmission.}
\label{fig:separation}
\end{figure}

\textbf{Scalability.} In Fig.~\ref{fig:scalability}, the gain rises from 0.61\% at $K{=}10$ to 3.34\% at 20 and 10.41\% at 50; all runs have $r_{\mathrm{BR},\infty}<4.6\!\times\!10^{-7}$~W. At $K{=}4$, active upper bounds make both strategies identical.

\textbf{Geometry Sensitivity.} Fig.~\ref{fig:separation} uses physical center-look separations. Over 20 realizations per point, NE gain falls from 10.75\% at $2^\circ$ to 0.75\% at $3^\circ$, 0.02\% at $5^\circ$, and effectively zero at $15^\circ$; mean $\rho(J_2)$ falls from 1.816 to 0.613, 0.023, and 0.001. Every run converges; the sufficient condition holds for 1/20 runs at $2^\circ$, 5/20 at $2.5^\circ$, and all runs from $3^\circ$ onward.

\textbf{Discussion.} The potential result is fixed-$K$ and weak-coupling; the spectral condition is sufficient. Both operators reuse one 250~MHz block with ideal intra-operator suppression. The patterns follow ITU-R S.1428-1 and S.1528-1, while the service region and co-cell association are reference assumptions. The conditional near-inline experiment quantifies a $2^\circ$--$3^\circ$ event, not its orbital probability. Thus, the results are reproducible link-budget evidence rather than a regulatory, vendor-specific, or time-resolved assessment.

\section{Conclusion}

We modeled equal-priority NGSO operators as a non-cooperative power-control game, proved pure-strategy NE existence, and derived distributed water-filling with a weighted spectral-norm convergence condition. Dual-ended off-axis discrimination yields a 0.06\% nominal NE gain, versus 3.91\% in the $2^\circ$--$3^\circ$ near-inline window and 10.41\% at $K{=}50$. The geometry sweep localizes useful adaptation, while centralized and lower-tail gaps expose the efficiency--fairness cost of independent sum-utility maximization. Minimum-rate-aware, pricing, and ephemeris-driven extensions remain future work.

\end{document}